\documentclass[lettersize,journal]{IEEEtran}
\usepackage{amsmath,amsfonts, amssymb}
\usepackage{algorithmic}
\usepackage{array}
\usepackage[caption=false,font=normalsize,labelfont=sf,textfont=sf]{subfig}
\usepackage{textcomp}
\usepackage{stfloats}
\usepackage{url}
\usepackage{verbatim}
\usepackage{graphicx}
\usepackage{cite}
\usepackage{hyperref}
\usepackage{enumitem}
\usepackage{xcolor}
\usepackage[linesnumbered,ruled,vlined]{algorithm2e}
\usepackage{bm}
\usepackage{amsthm}
\usepackage{booktabs}
\usepackage{siunitx}
\usepackage{threeparttable}

\newtheorem{theorem}{Theorem} 
\newtheorem{prop}{Proposition} 

\newcommand{\rr}{\mathbf{r}}
\newcommand{\dd}{\mathrm{d}}

\newcommand{\XX}{\mathbf{X}}

\newcommand{\EE}{\mathbb{E}}

\usepackage{soul}

\sethlcolor{yellow!40}
\newcommand{\rev}[1]{#1}

\soulregister\ref7
\soulregister\cite7
\soulregister\pageref7

\usepackage{multirow}
\newcommand{\twocols}[1]{\multicolumn{2}{c}{#1}}
\newcommand{\threecols}[1]{\multicolumn{3}{c}{#1}}
\newcommand{\fourcols}[1]{\multicolumn{4}{c}{#1}}

\newcommand{\tworowscol}[1]{\multicolumn{1}{c}{\multirow{2.5}{*}{#1}}}
\begin{document}

\title{\rev{Efficient FRW Transitions via Stochastic Finite Differences for Handling Non-Stratified Dielectrics}}

\author{Jiechen Huang and Wenjian Yu,~\IEEEmembership{Fellow,~IEEE}
\thanks{Manuscript received Feb. 28, 2025.}
\thanks{The authors are with Dept. Computer Science \& Tech., BNRist, State Key Laboratory of Cryptography and Digital Economy Security, Tsinghua University, Beijing 100084, China. (e-mail: yu-wj@tsinghua.edu.cn).}}

\markboth{Journal of \LaTeX\ Class Files,~Vol.~14, No.~8, August~2021}%
{Shell \MakeLowercase{\textit{et al.}}: A Sample Article Using IEEEtran.cls for IEEE Journals}


\maketitle

\begin{abstract}
The accuracy of floating-random-walk (FRW) based capacitance extraction stands only when the recursive FRW transitions are sampled unbiasedly according to surrounding dielectrics.
Advanced technology profiles, featuring complicated non-stratified dielectrics, challenge the accuracy of existing FRW transition \rev{schemes that approximate dielectrics with stratified or eight-octant} patterns.
In this work, we propose an algorithm named MicroWalk, \rev{enabling accurate FRW transitions} for arbitrary dielectrics \rev{while keeping high} efficiency. 
It is provably unbiased and equivalent to using transition probabilities solved by finite difference method, but at orders of magnitude lower cost (802$\times$ faster). An enhanced 3-D capacitance solver is developed with a hybrid strategy for complicated dielectrics, combining MicroWalk with the 
\rev{special treatment for the first transition cube and the}
analytical algorithm for stratified \rev{cubes}. Experiments on real-world \rev{structures} show that our solver \rev{achieves a significant accuracy advantage over existing FRW solvers, while preserving high efficiency.}
\end{abstract}


\section{Introduction}
For capacitance extraction under advanced process nodes, stricter accuracy requirements (typically $\!<\!5\%$ error) necessitate the use of 3-D \textit{field solvers} \cite{yu2014advanced}, which compute capacitances among conductors by solving the partial differential equation (PDE) of electrostatic interactions, i.e., \textit{Laplace's equation}.
The well-established methods include the finite difference method (FDM) and finite element method (FEM), etc \cite{jin2015theory}. 
\rev{In recent years,} the \textit{floating random walk} (FRW) method, a Monte Carlo approach for PDEs, is gaining attention \rev{for} capacitance extraction \cite{yu2013rwcap}. \rev{It obviates the need for spatial meshing or generating linear equations, and is thus highly parallelizable and scalable for large structures\cite{huang2025parallel}.}

Despite its superior parallelism and capacity, the FRW method faces difficulties in accurately modeling \rev{\textit{non-stratified and high-contrast}} \textit{dielectrics}  \cite{huang2024enhancing}. 
This weakness has become more pronounced,
as the advanced technology profiles involve conformal coatings, \rev{layout-dependent effects (LDEs)} or high-k materials.
These dielectrics affect electric field and capacitance, as reflected in Laplace's equation
\begin{equation}\label{eq:laplace}
    \nabla\cdot(\varepsilon\nabla \phi)\!=\!0 \! \Rightarrow \! \frac{\partial^2\phi}{\partial x^2}\!+\!\frac{\partial^2\phi}{\partial y^2}\!+\!\frac{\partial^2\phi}{\partial z^2}\!=\!0, \mathrm{at~interior~nodes,}
\end{equation}
where $\varepsilon$ is the spatially varying dielectric permittivity, and $\phi$ is the electric potential.
\rev{They} implicitly affect the FRW method by altering the transition probability, known as \textit{surface Green's function} (SGF) \cite{huang2024enhancing}\rev{,
whose correctness is critical to the unbiasedness of capacitance estimates.} 
\rev{A capacitance extraction task} typically entails millions of \rev{FRW} transitions \rev{(from transition cube's center to boundary), so performing them quickly with sufficient accuracy is demanded}.

The transition approaches in multi-dielectric environments are central in the research of the FRW method. 
Invoking FDM to compute the SGF in every step is accurate but impractically slow (see Fig. \ref{fig:framework}).
The \textit{de facto} workaround was to precompute SGF look-up tables for \rev{stratified dielectric configurations\cite{yu2013rwcap,zhang2016improved,yang2022reduce}, or to use an eight-octant model\cite{song2020floating} to approximately handle general dielectrics with homogenization}.
A recent study \cite{huang2024enhancing} proposed a fast analytical algorithm called AGF, which is only accurate for transition cubes with stratified dielectrics.
These \rev{techniques could still incur large error} 
for non-stratified and high-contrast dielectrics in advanced \rev{process technology}.
Another work\cite{visvardis2023deep} trained deep neural networks to predict SGFs for \rev{non-stratified} dielectrics, but \rev{with a noticeable performance drop. In\mbox{\cite[Table VII]{visvardis2023deep}}, the proposed method RW\_ML was, on average, 12$\times$ slower than RW\_NOLDE that entirely relies on precomputation and ignores non-stratified dielectrics, and 11$\times$ slower than RW\_OCT that uses precomputation and the eight-octant model\cite{song2020floating}.} It remains a challenge to efficiently perform unbiased FRW transitions for non-stratified dielectrics.

\begin{figure}[!t]
    \centering
      \setlength{\abovecaptionskip}{0cm}
      \setlength{\belowcaptionskip}{0cm}
    \includegraphics[width=\linewidth]{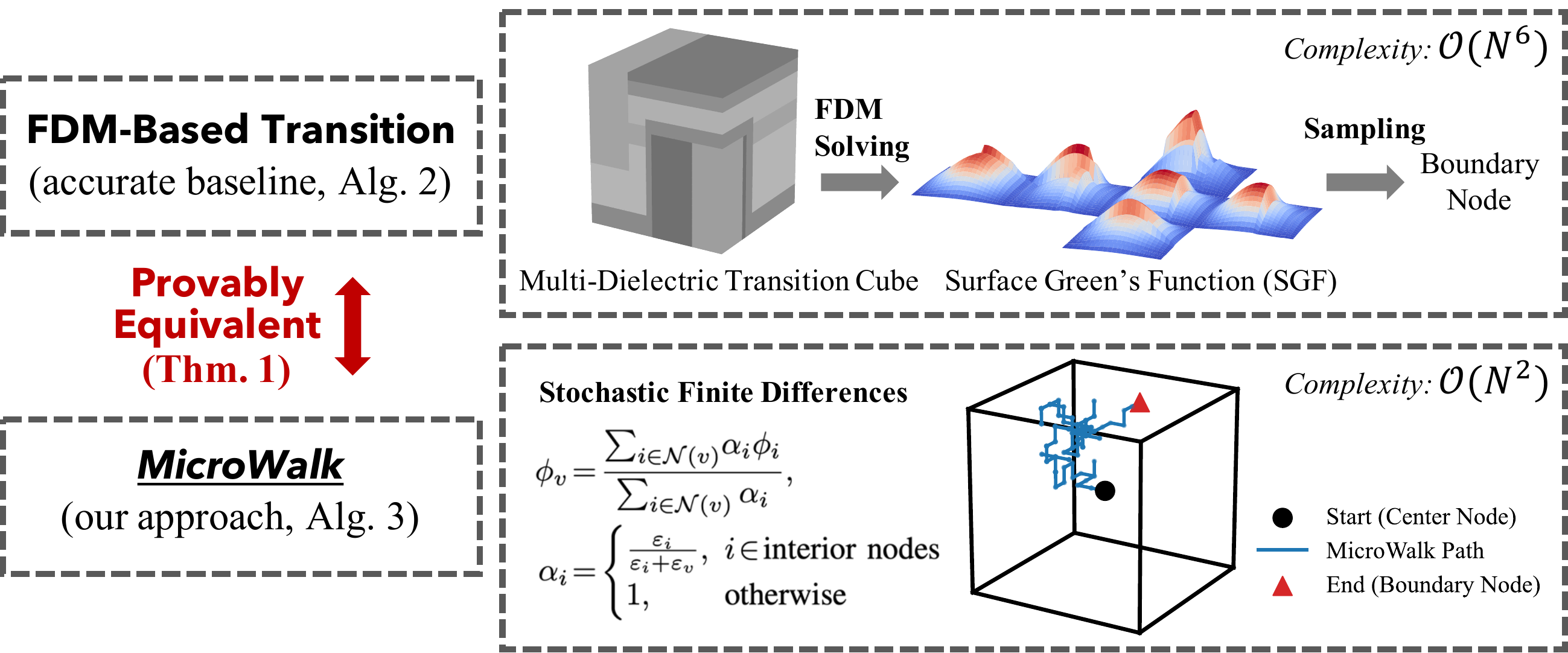}
    \caption{\rev{Two approaches of FRW transition for the transition cube with non-stratified dielectrics (depicted with different shade blocks). The proposed \textbf{MicroWalk} is as accurate as FDM, but with much lower computational cost.}}
    \label{fig:framework}
  \end{figure}
In this work, we \rev{tackle the defect of the FRW transition within} complicated dielectrics \rev{to} improve \rev{the FRW solver for} advanced technologies. Our main contributions include:
\begin{itemize}[leftmargin=*]
    \item A novel algorithm that performs unbiased FRW transitions without explicit SGF calculation, named \textbf{MicroWalk}, is derived based on \textit{stochastic finite differences} \cite{maire2016stochastic}. \rev{It enables accurate FRW transitions for any non-stratified transition cubes while preserving high efficiency.} 
    \item MicroWalk can be seen as the randomized version of FDM. We prove its \textit{equivalence} to sampling according to the SGFs computed by standard FDM (see Fig. \ref{fig:framework}). Furthermore, we show its time complexity is orders of magnitude lower than FDM, namely, $\mathcal{O}(N^2)$ vs. $\mathcal{O}(N^6)$. 
    \item A hybrid scheme with \rev{MicroWalk and AGF\cite{huang2024enhancing} is presented}.
\rev{The AGF with special treatment efficiently handles the first transition cube and the subsequent transitions with stratified dielectrics, and MicroWalk handles the subsequent transitions with more complicated dielectrics. The resulting} FRW solver accurately extracts \rev{capacitances in challenging} structures, \rev{for which} the state-of-the-art \cite{huang2024enhancing} produces wrong results. \rev{At the same time, this solver runs} $802\times$ faster than FDM-based transitions \rev{and costs comparable computational time to the  state-of-the-art} \cite{huang2024enhancing}.
\end{itemize}

\section{Preliminaries}
\subsection{FRW: Monte Carlo Approach for Capacitance Extraction}
\rev{The FRW method applies Monte Carlo integration to} capacitance extraction \cite{le1992stochastic}. It exploits the mean value property of Laplace's equation \rev{to solve the electric potential}:
\begin{gather}\label{eq:mvp}
    \phi(\rr) = \oint_{S(\rr)} P_{\varepsilon}(\rr_1|\rr) \phi(\rr_1) \dd s_1=\EE(\phi(\XX_{t+1}) | \XX_t=\rr),
\end{gather}
where $\{\XX_t,t\ge0 \}$ is a Markov chain with one-step transition probability $P_{\varepsilon}(\rr_1|\rr)$. In practice, $S(\rr)$ is a cubic surface centered at $\rr$.
$P_{\varepsilon}(\rr_1|\rr)$ is its \textit{surface Green's function} (SGF), determined by the distribution of permittivity $\varepsilon$ inside $S(\rr)$. 
By recursively following \eqref{eq:mvp} to perform transitions until $\phi(\XX_t)$ reaches a prescribed boundary condition, an unbiased estimate for $\phi(\XX_0)$ is obtained. This procedure, called \textit{FRW transition}, is the fundamental mechanism of the FRW method.

The capacitances of conductor $i$ can be represented by its induced charge $Q_i$ under certain voltage settings. Thus, utilizing Gauss's law and the above $\phi$-estimator,
a capacitance estimator can be derived.
We sketch the FRW algorithm as Alg.~\ref{algo:frw}, where 
\texttt{FirstTransition} corresponds to importance sampling and weight value calculation (see \cite[Alg. 3]{yu2013rwcap}) and \texttt{Transition} denotes the function for FRW transitions.
For a more theoretical setup and proofs, see \cite[Sec. III]{huang2024floating}.

\begin{algorithm}[h]
\begin{minipage}{\linewidth}
    \caption{FRW Algorithm for Capacitance Extraction.}\label{algo:frw}
\end{minipage}\KwIn{A 3-D structure of conductors and dielectrics; a master conductor $i$ and its Gaussian surface $G$.}
\KwOut{Capacitance values $C_{ij},\ \forall j$.}
$C_{ij} \gets 0,\ \forall j$; $n_{\text{path}} \gets 0$\;
\Repeat{\text{convergence criterion met}}{
    $\rr\gets \texttt{GaussianSurfaceSampling}(G)$\;
    $(\rr_1, \omega) \gets \texttt{FirstTransition}(\rr)$\;
    \Repeat{$\rr_1$ lands on a conductor surface}{
        $\rr_1 \gets \texttt{Transition}(\rr_1)$\;
    }
    $C_{ik} \gets C_{ik}+\omega$;\hfill \textcolor{blue}{// suppose $\rr_1$ is on conductor $k$;}\\ 
    $n_{\text{path}} \gets n_{\text{path}}+1$; 
}
$C_{ij} \gets \frac{C_{ij}}{n_{\text{path}}},\ \forall j$\;
\end{algorithm}

\subsection{Multi-Dielectric SGFs and Existing Approaches}\label{sec:sgf}
The SGFs are crucial as they define the probability distributions for FRW transitions. 
The unbiasedness of capacitance estimates holds only when the transitions are performed obeying correct random distributions. 
To facilitate this, transition cubes are discretized into $N\!\times\!N\!\times\!N$ uniform lattices, in which case \eqref{eq:mvp} becomes (only the center point needs estimation)
\begin{equation}\label{eq:discrete}
    \phi_{\text{center}} = \sum_{i=1}^{6N^2} \mathbf{P}_\varepsilon(\mathbf{x}_i)\bm\phi_B(\mathbf{x}_i)=\mathbb{E}_{\mathbf{x}\sim \mathbf{P}_\varepsilon}(\bm\phi_B(\mathbf{x})),
\end{equation}
where $\bm\phi_B$ denotes the potential of the $6N^2$ surface panels.
The FDM can compute the discrete SGF $\mathbf{P}_\varepsilon$ for general dielectric configurations with a time complexity of $\mathcal{O}(N^6)$, but this is too slow for online (real-time) computation.

\rev{For transition cubes with stratified (planar) dielectrics}, earlier works use offline FDM to precompute SGF look-up tables \rev{and apply dielectric homogenization to find the best match\cite{yu2013rwcap, zhang2016improved, yang2022reduce}. The AGF approach\cite{huang2024enhancing} enables} online computation for any stratified dielectrics, \rev{improving overall accuracy without efficiency loss}.
\textbf{For \rev{handling} non-stratified dielectrics,} most existing works homogenize them into solvable patterns such as eight-octant cubes \cite{song2020floating} or stratified cubes \cite{huang2024enhancing}.
However, these empirical approximations compromise accuracy. 
Another work uses deep neural networks to predict SGFs for general dielectrics. \rev{However, it considerably degrades} computational speed\rev{, compared to the FRW version with precomputed tables and octant model\mbox{\cite[Table VII]{visvardis2023deep}}.} 
\rev{Notice that the \texttt{FirstTransition} in Alg.~\ref{algo:frw} requires the calculation of weight value, whose variance is crucial to the convergence rate of Monte Carlo estimation. This makes stochastic finite differences unsuitable for handling the first transition. In this work, we focus on the accurate treatment of subsequent transition cubes (except the first one) with complex non-stratified dielectrics.}

\section{\rev{Improved FRW Solver with the Transitions Based on Stochastic Finite Differences}}

\subsection{Finite Difference Schemes for Generating SGFs}
\begin{figure}[h]
  \centering
    \setlength{\abovecaptionskip}{0cm}
	\setlength{\belowcaptionskip}{0cm}
  \includegraphics[width=0.85\linewidth]{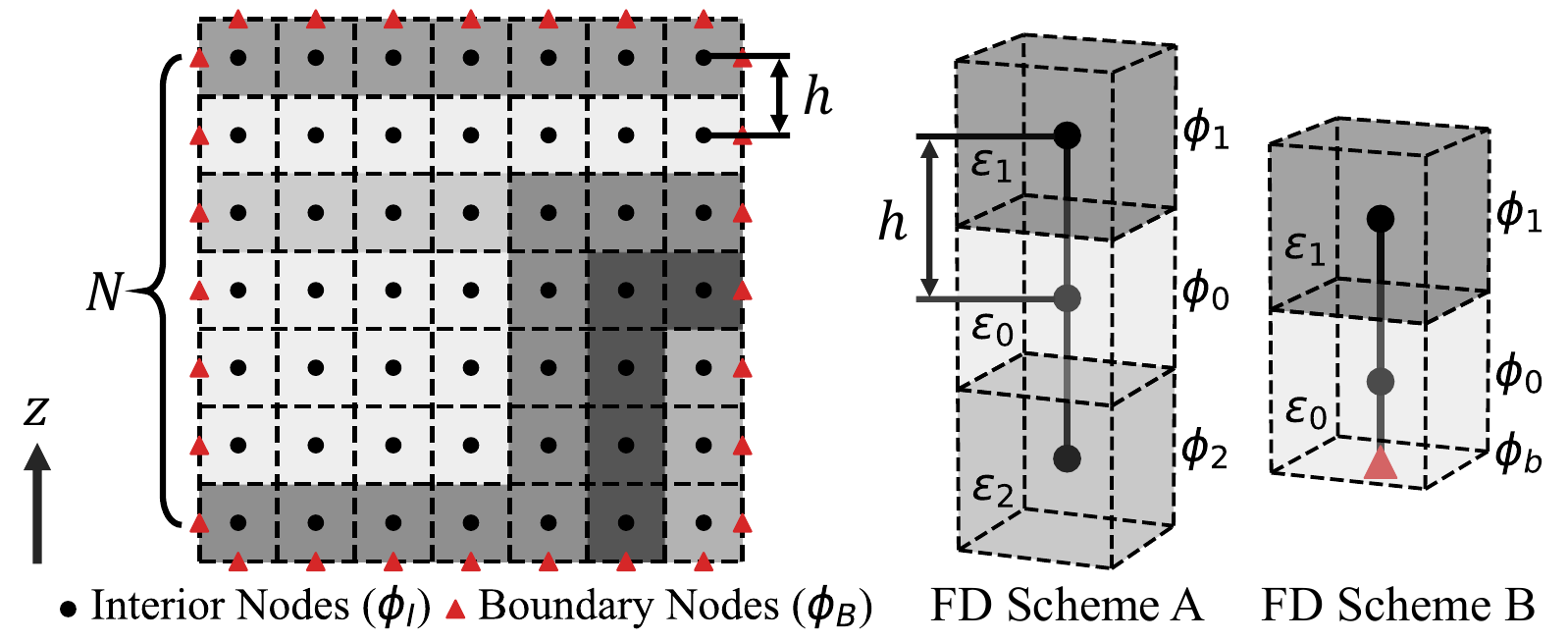}
  \caption{FDM for transition cube's SGF and two finite difference schemes.}
  \label{fig:fdm_scheme}
\end{figure}
Consider a voxelized transition cube (see its front cross-section in the left part of Fig. \ref{fig:fdm_scheme}), and assume the dielectric permittivity in each voxel is constant. 
$N^3$ interior nodes are allocated at the voxel centers and $6N^2$ boundary nodes are at surface panel centers.
At interior nodes, 
two finite difference schemes are utilized to \rev{convert \mbox{\eqref{eq:laplace}} to} linear equations, as depicted in the right part of Fig. \ref{fig:fdm_scheme}. 
Scheme A is for a node (with electric potential $\phi_0$) with two adjacent interior nodes ($\phi_1$ and $\phi_2$).
Let $\phi'_1$ ($\phi'_2$) be electric potential at the midpoint to node with potential $\phi_1$ ($\phi_2$). 
The second-order partial difference in \eqref{eq:laplace} becomes
\begin{equation}\label{eq:fdscheme1}
    \frac{\partial^2\phi}{\partial z^2} \approx \frac{\phi'_1+\phi'_2-2\phi_0}{h^2/4}=\frac{\frac{\varepsilon_1(\phi_1-\phi_0)}{\varepsilon_1+\varepsilon_0}+\frac{\varepsilon_2(\phi_2-\phi_0)}{\varepsilon_2+\varepsilon_0}}{h^2/4},
\end{equation}
where the second equality is based on the interface condition: $\varepsilon_1 (\phi_1-\phi'_1) =\varepsilon_0(\phi'_1-\phi_0)$. 
For Scheme B involving a boundary node, 
we similarly derive the partial difference:
\begin{equation}
    \frac{\partial^2\phi}{\partial z^2} \approx \frac{\frac{\varepsilon_1(\phi_1-\phi_0)}{\varepsilon_1+\varepsilon_0}+\phi_b-\phi_0}{h^2/4}.
\end{equation}
Combining these formulas (also those along $x$- and $y$-axis), we can rewrite \eqref{eq:laplace} for any interior node $v$ as 
\begin{equation}\label{eq:phi}
\phi_v\!=\!\frac{\sum_{i\in\mathcal{N}(v)}\!\alpha_{i}\phi_i}{\sum_{i\in\mathcal{N}(v)} \alpha_i},\ \alpha_i\!=\!
    \begin{cases}
        \frac{\varepsilon_i}{\varepsilon_i+\varepsilon_v}, & \hspace{-6pt} \text{node} ~i\!\in\!\text{interior\ nodes}\\
        1, & \hspace{-5pt} \text{otherwise}\\
    \end{cases},
\end{equation}
where $\mathcal{N}(v)$ is node $v$'s neighbor set. 
Combining these equations for all interior nodes, we obtain a linear equation system:
\begin{equation}\label{eq:mat}
    \mathbf{A}_{II} \bm{\phi}_I - \mathbf{A}_{IB} \bm{\phi}_B = \mathbf{0} ~ \Rightarrow ~  \bm{\phi}_I=\mathbf{A}_{II}^{-1}\mathbf{A}_{IB}\bm{\phi}_B,
\end{equation}
where $\bm{\phi}_I$ and $\bm{\phi}_B$ denote the potential at interior and boundary nodes, $\mathbf{A}_{II}\!\!\in\!\! \mathbb{R}^{N^3\!\times\! N^3}$, and $\mathbf{A}_{IB}\!\!\in\!\!\mathbb{R}^{N^3\!\times\! 6N^2}$. Let $\mathbf{e}_{c}$ be the indicator vector for the center node so that $\phi_{\text{center}}\!=\!\mathbf{e}^T_c\bm\phi_I$.
By comparing \eqref{eq:discrete} with \eqref{eq:mat}, we derive the discrete SGF:
\begin{equation}
\mathbf{P}_{\varepsilon}=\mathbf{e}_{c}^T\mathbf{A}_{II}^{-1}\mathbf{A}_{IB},
\end{equation}
which is a set of probabilities for sampling on $S$.
This FDM-based method can compute SGFs for arbitrary dielectric configurations, and is described as 
Alg. \ref{algo:fdm} in align with function \texttt{Transition} in Alg. \ref{algo:frw}. Since $\mathbf{A}_{II}$ is highly sparse, its time complexity is $\mathcal{O}(N^6)$ \cite{huang2024enhancing}.
\begin{algorithm}[h]
    \caption{The FDM-based \texttt{Transition}$(\rr_1)$.}
\label{algo:fdm}
$S \gets \text{build the transition cube centered at } \rr_1$\;
Compute $\mathbf{A}_{II}$ and $\mathbf{A}_{IB}$ in \eqref{eq:mat} according to the FD formulas \eqref{eq:fdscheme1}-\eqref{eq:phi} and the dielectric configuration in $S$\;
$\mathbf{P}_\varepsilon \gets \mathbf{e}_{c}^T\mathbf{A}_{II}^{-1}\mathbf{A}_{IB}$\;
\Return{a random point on $S$ sampled according to $\mathbf{P}_\varepsilon$}\;
\end{algorithm}
\subsection{\rev{MicroWalk: Transitioning} via Stochastic Finite Differences}
In this subsection, we derive an equivalent algorithm to Alg. \ref{algo:fdm}, based on the observation that
the coefficients of $\phi_i$ in \eqref{eq:phi} form a set of probabilities.
Thus, \eqref{eq:phi} implies an absorbing Markov chain \cite{grinstead2012introduction} for computing the electric potential on the finite difference lattice,
where the interior and boundary nodes are transient and absorbing states, respectively. The corresponding transition matrix is
\begin{equation}
    \mathbf{Q} = \begin{bmatrix}
        \mathbf{I}_{N^3}- \mathbf{A}_{II} & \mathbf{A}_{IB}\\
        \mathbf{O} & \mathbf{I}_{6N^2}\\
    \end{bmatrix},
\end{equation}
where $\mathbf{I}$ and $\mathbf{O}$ denote the identity matrix and zero matrix respectively.
The transitions following $\mathbf{Q}$ can be seen as \textit{stochastic finite differences} for Laplace's equation \cite{maire2016stochastic}.
They start at the center and repeat recursively until a boundary node is reached, which serves as the estimation point for $\phi_{\text{center}}$. This derives an implementation of function \texttt{Transition} as in Alg. \ref{algo:microwalk}, named the \textbf{MicroWalk} algorithm. 
\begin{algorithm}[h]
\begin{minipage}{0.95\linewidth}
    \caption{MicroWalk: an implementation of \texttt{Transition}$(\rr_1)$ based on stochastic finite differences.}\label{algo:microwalk}
\end{minipage}$S \gets \text{build the transition cube centered at } \rr_1$\;
$v \gets\text{center node of the finite difference lattice in }S$\;
\Repeat{$v\in$ boundary nodes}{
    Compute coefficients $\alpha_i,\ \forall i\in \mathcal{N}(v)$ in \eqref{eq:phi}\;
    $\beta_i\gets \frac{\alpha_i}{\sum_{i\in \mathcal{N}(v)} \alpha_i},\ \forall i \in \mathcal{N}(v)$; \hfill \textcolor{blue}{// normalization;}\\
    $v\gets\text{sample a node from } \mathcal{N}(v)$ according to $\{\beta_i\}$\;
}
\Return{the point corresponding to node $v$}\;
\end{algorithm}

The MicroWalk algorithm is innovative in that it does not require explicit SGF calculation while achieving unbiased FRW transitions.
It even avoids building the complete dielectric representations, which would already cost $\mathcal{O}(N^3)$ time.
It can be viewed as simulating a particle diffusion in the dielectric lattice.
In the following, Theorem \ref{thm:eq} proves its correctness and equivalence to the FDM-based transitions, while \rev{Proposition~\ref{thm:complexity}} establishes its efficiency with  $\mathcal{O}(N^2)$ time \rev{complexity. } 
\begin{theorem}\label{thm:eq}
    Alg. \ref{algo:fdm} and Alg. \ref{algo:microwalk} are equivalent, i.e., their output points follow the same distribution exactly.
\end{theorem}
\begin{proof}
    We need to show that the output point of Alg. \ref{algo:microwalk}, denoted $v^*$, follows the distribution of the discrete SGF $\mathbf{P}_\varepsilon$. 
    Let $\mathbf{B}$ denote the \textit{absorption probability matrix} of the absorbing chain $\mathbf{Q}$, where $\mathbf{B}_{ij}$ is the probability of absorption in state $j$ if starting from transient state $i$. By definition, $v^*\!\sim\! \mathbf{e}^T_c \mathbf{B}$. Directly following \cite[Theorem 11.6]{grinstead2012introduction}, we have
    \begin{equation}
        \mathbf{B}= \mathbf{A}_{II}^{-1}\mathbf{A}_{IB}  ~ \Rightarrow  ~ v^*\sim \mathbf{e}^T_c\mathbf{A}_{II}^{-1}\mathbf{A}_{IB}=\mathbf{P}_\varepsilon,
    \end{equation}
    which proves the original statement.
\end{proof}
\begin{prop}\label{thm:complexity}
    The time complexity of Alg. \ref{algo:microwalk} is about $\mathcal{O}(N^2)$.
\end{prop}
    Notice that line 4 in Alg. \ref{algo:microwalk} would query the dielectrics at 7 local nodes. 
    By preprocessing dielectric structures within current transition cube via efficient space management \cite{song2020floating}, and applying memorization to cache $\alpha_i$ values for reuse, we can ensure that \rev{executing line 4 costs $\mathcal{O}(1)$ time}. 
    Thus, the loop body is \rev{of $\mathcal{O}(1)$ complexity}.

    \rev{Denote the number of iterations in Alg. \ref{algo:microwalk} by $\tau$. Directly following\mbox{\cite[Theorem 11.5]{grinstead2012introduction}}, we can compute its expectation via:}
    \begin{equation}\label{eq:absorb}
        \mathbb{E}(\tau) = \mathbf{e}^T_c\mathbf{A}_{II}^{-1}\mathbf{1},
    \end{equation}
    \rev{where $\mathbf{1}$ is the all-one vector of length $N^3$.}
    \rev{For single-dielectric cubes, $\mathbf{A}_{II}$ follows a fixed structure, from which we can derive $\mathbb{E}(\tau)\approx 0.3373N^2$. For general dielectric configurations, the arbitrariness of $\mathbf{A}_{II}$ makes $\mathbb{E}(\tau)$ difficult to bound exactly, but our numerical experiments confirm that it is of $\mathcal O (N^2)$ as shown in Fig. \ref{fig:expand}(a). Therefore, we conclude that the time complexity of Alg.~\ref{algo:microwalk} is practically $\mathcal O(N^2)$.}

Furthermore, MicroWalk offers a distinct advantage: it enables larger hop lengths and faster stopping by allowing conductors inside transition cubes.
Existing works require conductor-free cubes for ease of SGF characterization \cite{yu2013rwcap,huang2025parallel,huang2024enhancing,visvardis2023deep,zhang2016improved,yang2022reduce,song2020floating}. This constraint is removed since MicroWalk can adapt stochastic finite difference schemes for intrusive conductors, as shown in Fig. \ref{fig:expand}(b).
This version of \underline{MicroWalk} with \underline{e}xpanded transition cubes is named \textbf{MicroWalk-E}. Note that due to the coarser finite difference lattice, MicroWalk-E may trade off some accuracy. 
\rev{The expansion factor of $5\times$ is empirically chosen in experiments (see Sec. \ref{sec:result}), which show that the caused error remains insignificant.}

\begin{figure}[t]
    \centering
      \setlength{\abovecaptionskip}{0cm}
      \setlength{\belowcaptionskip}{0cm}
    \includegraphics[width=\linewidth]{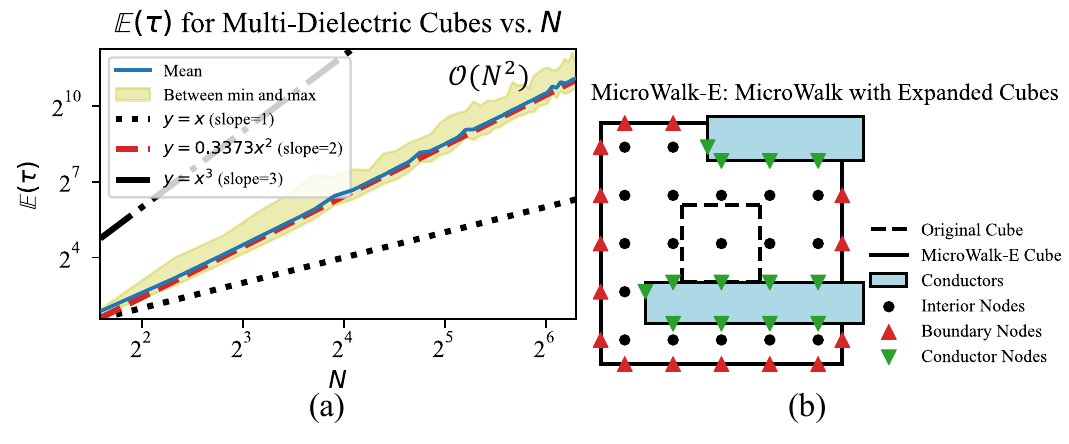}
    \caption{(a) For each $N$, we generate 500 cubes enclosing randomly-placed dielectric blocks with $\varepsilon\!\sim\!\text{Exp}(0.1)$\rev{, an exponential distribution,} and compute their $\mathbb{E}(\tau)$ via \eqref{eq:absorb}. The min/max/mean \rev{values of $\mathbb{E}(\tau)$} are plotted on log-scale, revealing a clear $\mathcal O(N^2)$ scaling. (b) The MicroWalk-E variant allows transition cubes to contain conductors and adjusts the lattice accordingly.}
    \label{fig:expand}
\end{figure}

\subsection{\rev{Hybrid Strategies for Achieving Higher Efficiency}}\label{sec:hybrid}
\rev{In terms of time complexity, MicroWalk is less efficient than AGF\cite{huang2024enhancing} (see Table~\ref{tab:comparison}).}
As detailed in \cite{huang2024enhancing}, when paired with a hash-table cache, AGF can obtain SGFs in effectively $\mathcal O(1)$ time \rev{due to the high hit rate} (see \cite[Sec. 5.2]{huang2024enhancing}).
Considering this, we devise a \textit{hybrid strategy} combining AGF and MicroWalk for higher \rev{efficiency.
For the FRW transitions other than the first, stratified cubes are handled via AGF with hash-table cache, and non-stratified cubes are handled via the proposed MicroWalk. Note that the deep-learning-driven approximation in\cite{visvardis2023deep} is expensive, for which constructing the input of neural network already costs $\mathcal O(N^3)$ time.} 

\begin{table}[h]
    \centering
    \setlength{\abovecaptionskip}{0cm}
	\setlength{\belowcaptionskip}{0cm}
    \caption{Comparisons of three \rev{accurate FRW transition} approaches.}
    \label{tab:comparison}
    \setlength\tabcolsep{1pt}
    \begin{tabular}{l ccc}
    \toprule
     & FDM & AGF \cite{huang2024enhancing} & MicroWalk \\ 
    \midrule
    Complexity for SGF & $\mathcal O(N^6)$ & $\mathcal O(1)$/$\mathcal O(N^3)$ for hit/miss & - \\
    Complexity for Sampling & $\mathcal O(\log N)$ & $\mathcal O(\log N)$ & $\mathcal O(N^2)$ \\
    Dielectric Profile & Any & Stratified dielectrics & Any \\
    \bottomrule
    \end{tabular}
\end{table}

\rev{For the first transition in a walk, which requires SGF gradient sampling and weight value calculation, the problem is more involved. AGF is applied to stratified cubes,
but stochastic finite difference is not suitable for computing weight value since it would considerably increase the variance. 
So, we handle the non-stratified cubes encountered in the first step via the AGF-based special treatment as per in\cite{huang2024enhancing}, whose flowchart is shown as Fig.~\ref{fig:transition_flowchart}. The accuracy of this special treatment has been validated through our extensive experiments across various process technologies.
This strategy for the first transition was not presented in such detail in\cite{huang2024enhancing}.} 

\begin{figure}[h]
  \centering
    \setlength{\abovecaptionskip}{0cm}
	\setlength{\belowcaptionskip}{0cm}
  \includegraphics[width=\linewidth]{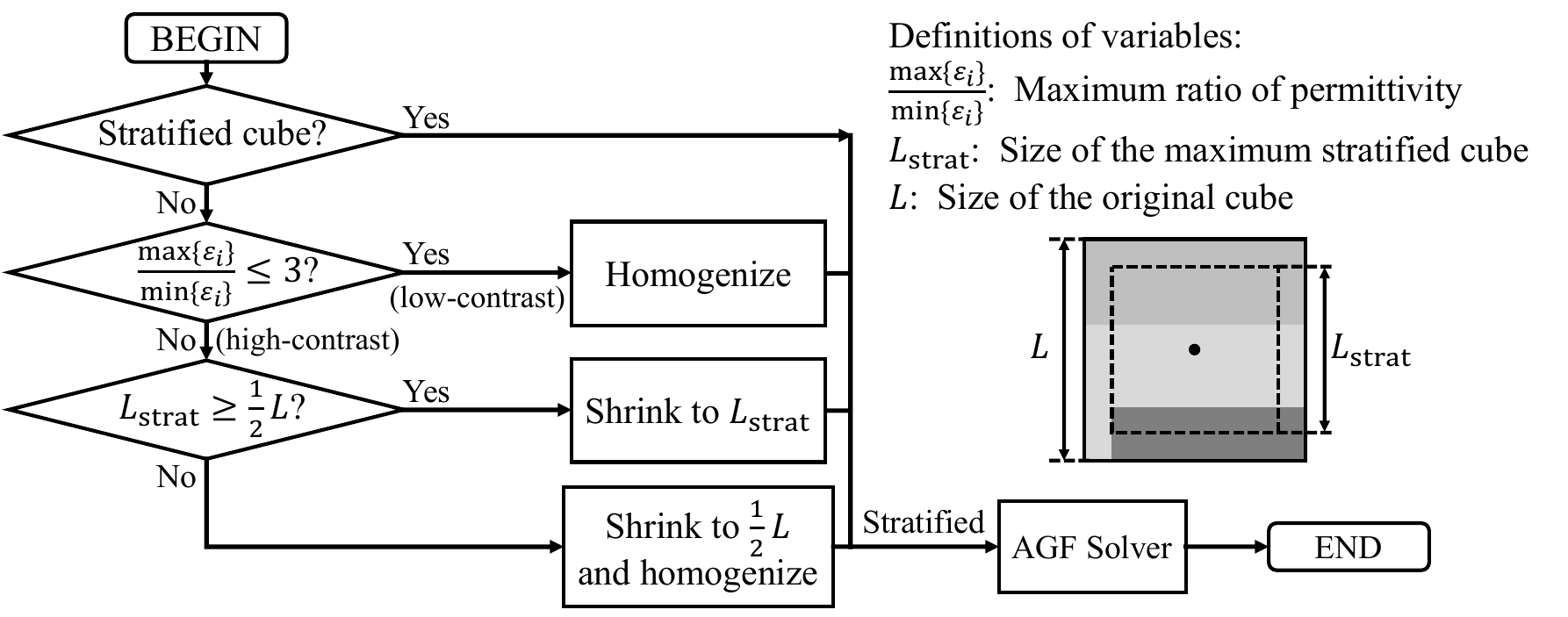}
  \caption{\rev{The special hybrid treatment for the first transition, as per in\cite{huang2024enhancing}.}}
  \label{fig:transition_flowchart}
\end{figure}
\section{Numerical Results}

We have implemented the proposed MicroWalk (Alg. \ref{algo:microwalk}) and MicroWalk-E (with $5\times$ cube expansion) in C++ based on our codes for \cite{huang2024enhancing}. 
Collaborating the hybrid strategy in Sec. \ref{sec:hybrid}, we obtain capacitance solvers \textbf{FRW-MW} and \textbf{FRW-MWE}.
The baselines include \textbf{FRW-FDM} and \textbf{FRW-AGF} \cite{huang2024enhancing}. Except for the \texttt{Transition} step, \rev{the above} FRW solvers share the same implementation and adopt $N\!=\!24$ as \cite{huang2024enhancing}. \rev{In addition, a precomputation-based FRW solver \textbf{RWCap4}\cite{song2020floating,yang2022reduce}, and} commercial tools Raphael and QuickCap are also tested for comparison. \rev{More details about the tested methods are in Table \ref{tab:baselines}. RWCap4 and FRW-AGF can be downloaded from\cite{numbda}.}

\begin{table}[h]
    \centering
    \setlength{\abovecaptionskip}{0cm}
	\setlength{\belowcaptionskip}{0cm}
    \begin{threeparttable}

    \caption{\rev{Comparisons of the tested FRW solvers.}}
    \label{tab:baselines}
    \setlength\tabcolsep{2pt}
    \begin{tabular}{l ccc}
    \toprule
    \tworowscol{Method} & \tworowscol{\texttt{FirstTransition}} & \twocols{\texttt{Transition}} \\
    \cmidrule(lr){3-4} 
    \multicolumn{1}{c}{} & \multicolumn{1}{c}{} & Stratified & Non-stratified \\
    \midrule
    RWCap4 \cite{song2020floating, yang2022reduce} & Precomput. & Precomput. & Eight-octant \\
    FRW-AGF \cite{huang2024enhancing} & Special hybrid (Fig.~\ref{fig:transition_flowchart}) & AGF & Special hybrid w/ Oct.\tnote{a} \\
    FRW-FDM & Special hybrid (Fig.~\ref{fig:transition_flowchart}) & AGF & FDM (Alg. \ref{algo:fdm}) \\
    FRW-MW & Special hybrid (Fig.~\ref{fig:transition_flowchart}) & AGF & MicroWalk (Alg. \ref{algo:microwalk}) \\
    FRW-MWE & Special hybrid (Fig.~\ref{fig:transition_flowchart}) & AGF & MicroWalk-E \\
    \bottomrule
    \end{tabular}
    \begin{tablenotes}
      \item[a]Similar to Fig.~\ref{fig:transition_flowchart} but with the ``Homogenize'' branch replaced by ``Approximate with the eight-octant model for on-the-fly sampling''.
    \end{tablenotes}
    \end{threeparttable}
    
\end{table}

All test cases are real-world \rev{conductor} structures from our industrial partners.
Regression test cases RT1-RT8 are namely Cases 1-8 in \cite{huang2024enhancing}. As for technology nodes, RT1-RT3 and RT7-RT8 are at 16nm, RT4 at 28nm, RT5-RT6 at 55nm. Two new Cases 1-2 are FinFET structures at 12nm node. \rev{From Fig.~\ref{fig:cross} we see that they involve different kinds of} non-stratified dielectrics, with Cases 1-2 being the most complicated due to high-k dielectrics ($\varepsilon_{\text{high}}\!=\!22$).
For RT1-RT8, the FRW convergence criterion follows \cite[Sec. 5.2]{huang2024enhancing}. For the others, the standard deviation is set as $\!<\!1\%$ on self and the largest coupling capacitances.
Experiments are conducted with serial computing on Intel Xeon 8375C CPU. 

\begin{figure}[h]
  \centering
    \setlength{\abovecaptionskip}{0cm}
	\setlength{\belowcaptionskip}{0cm}
  \includegraphics[width=0.95\linewidth]{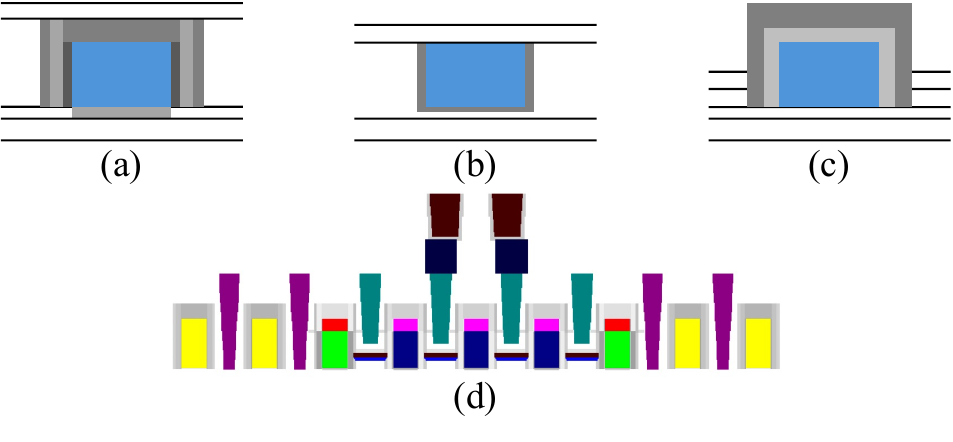}
  \caption{\rev{Simplistic cross-sections of the test cases, where dielectric blocks are in different shades of grey. (a) Cases RT1 and RT3, (b) Case RT2, (c) Cases RT4-RT6. (d) Cases 1-2 (interfaces of planar dielectrics not shown).}}
  \label{fig:cross}
\end{figure}

\subsection{Comparison Results of Capacitance Extraction}\label{sec:result}
The results (mainly self capacitance $C_{\text{self}}$ and runtime $T_{\text{total}}$) obtained with \rev{the different FRW solvers} are listed in Table \ref{tab:acc1}, along with the reference values \rev{from} Raphael and Quickcap. The average error of capacitances is defined as
$\text{Err}_{\text{avg}}=\frac{\sum_{\text{all caps}}\!|C_{ij}^{\text{(ours)}}-C_{ij}^{\text{(ref)}}|}{\sum_{\text{all caps}}|C_{ij}^{\text{(ref)}}|}.$
From Table \ref{tab:acc1} we see that FRW-FDM, FRW-MW, and FRW-MWE share similar accuracy, while \rev{RWCap4 and} FRW-AGF \rev{yield} wrong results on \rev{some challenging cases}. It takes 2349s, 132hr, 522s, 1883s, and 595s to finish all the extractions for QuickCap, FRW-FDM, FRW-AGF, FRW-MW, and FRW-MWE, respectively. FRW-MWE runs $3.9\times$, $802\times$ and $3.2\times$ faster than QuickCap, FRW-FDM and FRW-MW respectively. It is only $12\%$ slower than FRW-AGF, which is acceptable given the accuracy improvement. 

\rev{For the proposed FRW-MW solver, across all experiments, 79.9\% of first transition cubes are of stratified dielectrics, corresponding to the first branch in Fig.~\ref{fig:transition_flowchart}. The second, third, and fourth branches account for 9.3\%, 7.3\%, and 3.5\% of first transition cubes, respectively. For transition cubes other than the first, 70.0\% are stratified and 30.0\% are non-stratified.}

\begin{table*}[!t]
    \centering
	\setlength{\abovecaptionskip}{0pt}
    \setlength{\belowcaptionskip}{0pt}
    \caption{The computational results of proposed MicroWalk-based methods and the baselines, including the time cost of MicroWalk process ($T_{\mathrm{MW}}$). \rev{The  minimum $\mathrm{Err}_{\mathrm{avg}}$ is marked, and the $\mathrm{Err}_{\mathrm{avg}}$ larger than 5\% is underscored.} $C_{\mathrm{self}}$ is in unit of $10^{-15}$ F.
    }
    \label{tab:acc1}
    \setlength\tabcolsep{3pt}
\begin{tabular}{c c cc ccc ccc ccc cccc cccc}
    \toprule
    \tworowscol{Case} & Raphael & \twocols{QuickCap} & \threecols{\rev{RWCap4\cite{song2020floating,yang2022reduce}}} & \threecols{FRW-FDM} & \threecols{FRW-AGF \cite{huang2024enhancing}} & \fourcols{FRW-MW (Ours)} & \fourcols{FRW-MWE (Ours)} \\
    \cmidrule(lr){2-2} \cmidrule(lr){3-4} \cmidrule(lr){5-7} \cmidrule(lr){8-10} \cmidrule(lr){11-13} \cmidrule(lr){14-17} \cmidrule(lr){18-21}
    \multicolumn{1}{c}{} & $C_{\text{self}}$ & $C_{\text{self}}$ & $T_{\text{total}}$ & \rev{$C_{\text{self}}$} & \rev{$\text{Err}_{\text{avg}}$} & \rev{$T_{\text{total}}$} & $C_{\text{self}}$ & $\text{Err}_{\text{avg}}$ & $T_{\text{total}}$ & $C_{\text{self}}$ & $\text{Err}_{\text{avg}}$ & $T_{\text{total}}$ & $C_{\text{self}}$ & $\text{Err}_{\text{avg}}$ & $T_{\text{MW}}$ &$T_{\text{total}}$ & $C_{\text{self}}$ & $\text{Err}_{\text{avg}}$ & $T_{\text{MW}}$ &  $T_{\text{total}}$ \\
    \midrule
RT1 & 15.92 & 15.80 & 397s & 14.31 & \underline{10\%} & 118s & 15.86 & 0.4\% & 24hr & 15.84 & \textbf{0.4\%} & 65.7s& 15.82 & 0.6\% & 315s & 342s& 16.09 & 1.1\% & 55.9s & 67.8s\\
RT2 & 16.46 & 16.91 & 149s & 17.11 & 4.7\% & 19.8s & 16.70 & 1.7\% & 3.3hr & 16.72 & 1.7\% & 11.4s& 16.71 & 1.6\% & 40.6s & 49.1s& 16.67 & \textbf{1.4\%} & 7.2s & 12.5s\\
RT3 & 21.58 & 21.40 & 678s & 20.81 & 3.5\% & 34.9s & 21.21 & 1.9\% & 22hr & 21.12 & 2.3\% & 81.6s& 21.19 & \textbf{1.8\%} & 269s & 302s& 22.15 & 3.2\% & 41.6s & 57.9s\\
RT4 & 8.210 & 8.142 & 160s & 8.336 & 4.4\% & 105s & 8.226 & 0.2\% & 24hr & 8.305 & 1.2\% & 64.5s& 8.172 & \textbf{0.5\%} & 267s & 307s& 8.215 & 0.6\% & 56.1s & 82.1s\\
RT5 & 14.19 & 14.13 & 288s & 14.41 & 3.5\% & 118s & 14.06 & 0.8\% & 19hr & 14.25 & 0.7\% & 78.4s& 14.05 & 0.9\% & 215s & 265s& 14.15 & \textbf{0.3\%} & 51.0s & 89.7s\\
RT6 & 9.584 & 9.571 & 205s & 10.03 & \underline{6.7\%} & 89.0s & 9.577 & 0.6\% & 23hr & 9.835 & 3.1\% & 61.0s& 9.535 & \textbf{0.9\%} & 254s & 291s& 9.649 & 1.4\% & 60.1s & 84.4s\\
RT7 & 0.933 & 0.894 & 116s & 0.852 & \underline{7.0\%} & 18.0s & 0.890 & 3.3\% & 1.1hr & 0.911 & \textbf{1.0\%} & 16.5s& 0.903 & 2.0\% & 13.5s & 26.7s& 0.902 & 2.3\% & 4.4s & 16.5s\\
RT8 & 0.883 & 0.857 & 61s & 0.816 & \underline{6.8\%} & 10.0s & 0.880 & 0.5\% & 0.61hr & 0.868 & 0.9\% & 10.0s& 0.878 & \textbf{0.6\%} & 7.0s & 14.3s& 0.882 & \textbf{0.6\%} & 2.1s & 8.7s\\
1 & 0.188 & 0.187 & 80s & 0.179 & \underline{9.3\%} & 60.7s & 0.188 & 0.4\% & 8.5hr & 0.219 & \underline{20\%} & 69.3s& 0.186 & \textbf{0.8\%} & 92.4s & 150s& 0.185 & 2.0\% & 36.2s & 85.5s\\
2 & - & 1.241 & 215s & 1.195 & \underline{7.8\%} & 57.3s & 1.243 & 0.9\% & 7.2hr & 1.544 & \underline{24\%} & 63.7s& 1.248 & \textbf{1.6\%} & 70.0s & 136s& 1.218 & 1.8\% & 31.1s & 90.3s\\
    \bottomrule
    \end{tabular}
\end{table*}

\subsection{Case Study with Varying Parameters}
We conduct two parametric studies on Case 1 to explore MicroWalk's behaviors: 1) We vary $N$ to observe the scaling of runtime (Fig. \ref{fig:case1}(a)). MicroWalk's runtime grows quadratically, which verifies Proposition~\ref{thm:complexity}, and dominates the total runtime. 2) We further increase the $\varepsilon_{\text{high}}$ of high-k dielectrics to examine the capacitances solved by these methods (Fig.~\ref{fig:case1}(b)). The 
results confirm that FRW-MW maintains accuracy and robustness across a broad range of permittivity, whereas FRW-AGF's approximation induces substantial systematic errors.
\begin{figure}[h]
    \centering
      \setlength{\abovecaptionskip}{0cm}
      \setlength{\belowcaptionskip}{0cm}
    \includegraphics[width=1\linewidth]{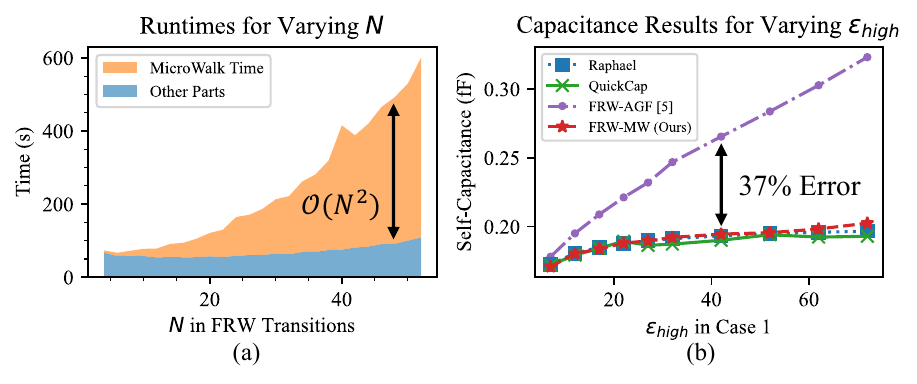}
    \caption{Parametric studies on Case 1: (a) Performance scaling: runtimes of FRW-MW for varying $N$, split into those for MicroWalk and other parts; (b) Robustness test: capacitance results for varying $\varepsilon_{\text{high}}$, across various methods.}
    \label{fig:case1}
\end{figure} 

\section{Conclusion}

In this work, we propose MicroWalk, a stochastic finite difference algorithm for arbitrary-dielectric FRW transitions. With it, we have enhanced the state-of-the-art 3-D capacitance solver FRW-AGF \cite{huang2024enhancing} for the accuracy on handling complicated non-stratified dielectrics, and developed the solver named FRW-MW and its variant FRW-MWE.
Experiment results validate the accuracy of them on structures under advanced technologies, in comparison with commercial golden tools. While for the structures where FRW-AGF yields errors above 20\%, the proposed FRW-MWE exhibits good accuracy with $\le$2\% error at the cost of a marginal increase of runtime.

\bibliographystyle{IEEEtran}
\bibliography{main}

\begin{thebibliography}{10}
\providecommand{\url}[1]{#1}
\csname url@samestyle\endcsname
\providecommand{\newblock}{\relax}
\providecommand{\bibinfo}[2]{#2}
\providecommand{\BIBentrySTDinterwordspacing}{\spaceskip=0pt\relax}
\providecommand{\BIBentryALTinterwordstretchfactor}{4}
\providecommand{\BIBentryALTinterwordspacing}{\spaceskip=\fontdimen2\font plus
\BIBentryALTinterwordstretchfactor\fontdimen3\font minus \fontdimen4\font\relax}
\providecommand{\BIBforeignlanguage}[2]{{%
\expandafter\ifx\csname l@#1\endcsname\relax
\typeout{** WARNING: IEEEtran.bst: No hyphenation pattern has been}%
\typeout{** loaded for the language `#1'. Using the pattern for}%
\typeout{** the default language instead.}%
\else
\language=\csname l@#1\endcsname
\fi
#2}}
\providecommand{\BIBdecl}{\relax}
\BIBdecl

\bibitem{yu2014advanced}
W.~Yu and X.~Wang, \emph{Advanced Field-Solver Techniques for RC Extraction of Integrated Circuits}.\hskip 1em plus 0.5em minus 0.4em\relax Springer, 2014.

\bibitem{jin2015theory}
J.~Jin, \emph{Theory and Computation of Electromagnetic Fields, Second Edition}.\hskip 1em plus 0.5em minus 0.4em\relax Wiley, 2015.

\bibitem{yu2013rwcap}
W.~Yu, H.~Zhuang, C.~Zhang, G.~Hu, and Z.~Liu, ``{RWCap: A floating random walk solver for 3-D capacitance extraction of very-large-scale integration interconnects},'' \emph{IEEE Trans. Comput.-Aided Des. Integr. Circuits Syst.}, vol.~32, no.~3, pp. 353--366, 2013.

\bibitem{huang2025parallel}
J.~Huang, S.~Liu, and W.~Yu, ``A parallel floating random walk solver for reproducible and reliable capacitance extraction,'' in \emph{Proc. DATE}, 2025.

\bibitem{huang2024enhancing}
J.~Huang and W.~Yu, ``Enhancing {3-D} random walk capacitance solver with analytic surface {Green's} functions of transition cubes,'' in \emph{Proc. DAC}, 2024.

\bibitem{zhang2016improved}
B.~Zhang, W.~Yu, and C.~Zhang, ``Improved pre-characterization method for the random walk based capacitance extraction of multi-dielectric vlsi interconnects,'' \emph{Int. J. Numer. Model.: Electron. Netw. Devices Fields}, vol.~29, no.~1, pp. 21--34, 2016.

\bibitem{yang2022reduce}
M.~Yang, W.~Yu, M.~Song, and N.~Xu, ``Volume reduction and fast generation of the precharacterization data for floating random walk-based capacitance extraction,'' \emph{IEEE Trans. Comput.-Aided Des. Integr. Circuits Syst.}, vol.~41, no.~5, pp. 1467--1480, 2022.

\bibitem{song2020floating}
M.~Song, M.~Yang, and W.~Yu, ``Floating random walk based capacitance solver for {VLSI} structures with non-stratified dielectrics,'' in \emph{Proc. DATE}, 2020, pp. 1133--1138.

\bibitem{visvardis2023deep}
M.~Visvardis, P.~Liaskovitis, and E.~Efstathiou, ``Deep-learning-driven random walk method for capacitance extraction,'' \emph{IEEE Trans. Comput.-Aided Des. Integr. Circuits Syst.}, vol.~42, no.~8, pp. 2643--2656, 2023.

\bibitem{maire2016stochastic}
S.~Maire and G.~Nguyen, ``Stochastic finite differences for elliptic diffusion equations in stratified domains,'' \emph{Mathematics and Computers in Simulation}, vol. 121, pp. 146--165, 2016.

\bibitem{le1992stochastic}
{Y. L. {Le Coz} and R. B. Iverson}, ``{A stochastic algorithm for high speed capacitance extraction in integrated circuits},'' \emph{Solid-State Electron.}, vol.~35, no.~7, pp. 1005--1012, 1992.

\bibitem{huang2024floating}
J.~Huang, M.~Yang, and W.~Yu, ``The floating random walk method with symmetric multiple-shooting walks for capacitance extraction,'' \emph{IEEE Trans. Comput.-Aided Des. Integr. Circuits Syst.}, vol.~43, no.~7, pp. 2098--2111, 2024.

\bibitem{grinstead2012introduction}
C.~Grinstead and J.~Snell, ``Absorbing markov chains,'' in \emph{Introduction to Probability, Second Edition}.\hskip 1em plus 0.5em minus 0.4em\relax American Mathematical Society, 2012, ch. 11.2, pp. 415--432.

\bibitem{numbda}
{Numbda Group,} \emph{RWCap-v4 and RWCap-v5}. Accessed: Jun. 2025. [Online]. Available: \url{https://numbda.cs.tsinghua.edu.cn/download.html}.

\end{thebibliography}


\end{document}